\documentclass[11pt]{article}
\usepackage[margin=1in]{geometry}
\usepackage{complexity}
\usepackage{mathtools,amsmath,amssymb}
\usepackage{xspace}
\RequirePackage[colorlinks=true]{hyperref}
\hypersetup{
  linkcolor=[rgb]{0,0,0.5},
  citecolor=[rgb]{0, 0.5, 0},
  urlcolor=[rgb]{0.5, 0, 0}
}
\usepackage{dsfont}

\usepackage{amsthm}
\usepackage{thmtools,thm-restate}

\numberwithin{equation}{section}
\declaretheoremstyle[bodyfont=\it,qed=\qedsymbol]{noproofstyle}

\declaretheorem[name=Observation,numbered=no]{observation*}

\declaretheorem[numberlike=equation]{theorem}

\declaretheorem[name=Theorem,numbered=no]{theorem*}

\declaretheorem[numberlike=equation]{lemma}
\declaretheorem[name=Lemma,numbered=no]{lemma*}

\declaretheorem[name=Corollary,numbered=no]{corollary*}

\declaretheorem[name=Proposition,numbered=no]{proposition*}

\declaretheorem[name=Claim,numbered=no]{claim*}

\declaretheorem[name=Conjecture,numbered=no]{conjecture*}

\declaretheorem[numberlike=equation]{question}
\declaretheorem[name=Question,numbered=no]{question*}

\declaretheoremstyle[bodyfont=\it,qed=$\lozenge$]{defstyle}

\declaretheorem[unnumbered,name=Definition,style=defstyle]{definition*}

\declaretheorem[unnumbered,name=Example,style=defstyle]{example*}

\declaretheorem[unnumbered,name=Notation=defstyle]{notation*}

\declaretheorem[unnumbered,name=Construction,style=defstyle]{construction*}

\declaretheorem[unnumbered,name=Remark,style=defstyle]{remark*}

\usepackage{nth}
\usepackage{intcalc}
\usepackage{etoolbox}
\usepackage{xstring}
\hypersetup{
}

\usepackage{ifpdf}
\ifpdf
\else
\usepackage[quadpoints=false]{hypdvips}
\fi

\newcommand{\ECCCurl}[2]{http://eccc.weizmann.ac.il/report/\ifnumcomp{#1}{>}{93}{19}{20}#1/#2/}

\newcommand{\shortECCC}[2]{\texttt{\href{http://eccc.weizmann.ac.il/report/\ifnumcomp{#1}{>}{93}{19}{20}#1/#2/}{eccc:TR#1-#2}}}

\newcommand{\parseECCC}[1]{
\StrSubstitute{#1}{TR}{}[\tmpstring]%
\IfSubStr{\tmpstring}{/}{ 
\StrBefore{\tmpstring}{/}[\ecccyear]%
\StrBehind{\tmpstring}{/}[\ecccreport]%
}{
\StrBefore{\tmpstring}{-}[\ecccyear]%
\StrBehind{\tmpstring}{-}[\ecccreport]%
}%
\shortECCC{\ecccyear}{\ecccreport}}

	\renewcommand{\vec}[1]{{\mathbf{#1}}}

	\makeatletter
	\newcommand{\va}{{\vec{a}}\@ifnextchar{^}{\!\:}{}}
	\newcommand{\vb}{{\vec{b}}\@ifnextchar{^}{\!\:}{}}
	\newcommand{\vc}{{\vec{c}}\@ifnextchar{^}{\!\:}{}}
	\newcommand{\vd}{{\vec{d}}\@ifnextchar{^}{\!\:}{}}
	\newcommand{\ve}{{\vec{e}}\@ifnextchar{^}{\!\:}{}}
	\newcommand{\vy}{{\vec{y}}\@ifnextchar{^}{\!\:}{}}
	\newcommand{\vs}{{\vec{s}}\@ifnextchar{^}{\!\:}{}}
	\newcommand{\vt}{{\vec{t}}\@ifnextchar{^}{\!\:}{}}
	\newcommand{\vu}{{\vec{u}}\@ifnextchar{^}{\!\:}{}}
	\newcommand{\vx}{{\vec{x}}\@ifnextchar{^}{}{}}		
	\newcommand{\vz}{{\vec{z}}\@ifnextchar{^}{\!\:}{}}

	\newcommand{\vY}{{\vec{Y}}\@ifnextchar{^}{\!\:}{}}
	\newcommand{\vX}{{\vec{X}}\@ifnextchar{^}{}{}}		
	\newcommand{\vZ}{{\vec{Z}}\@ifnextchar{^}{\!\:}{}}
	\newcommand{\vG}{{\vec{G}}\@ifnextchar{^}{\!\:}{}}
	
	\makeatother

\newcommand{\F}{\mathbb{F}}
\newcommand{\N}{\mathbb{N}}

\newcommand{\Comp}{\mathbb{C}}

\newcommand{\set}[1]{\left\{#1\right\}}

\newcommand{\sps}{\sum\prod\sum}

\newcommand{\sym}{\mathrm{SYM}}

\newcommand{\cf}{\mathrm{Coeff}}
\def\epsilon{\varepsilon}

\date{}

\title{On top fan-in vs formal degree for depth-$3$ arithmetic circuits}
\author{
Mrinal Kumar\thanks{Center for Mathematical Sciences and Applications, Harvard University, Cambridge, Massachusetts, USA. Email: \texttt{mrinalkumar08@gmail.com}. }}
\begin{document}
\maketitle

\begin{abstract}
We show that over the field of complex numbers, \emph{every} homogeneous polynomial of degree $d$ can be approximated (in the border complexity sense) by a depth-$3$ arithmetic circuit of top fan-in at most $d+1$. This is quite surprising since there exist homogeneous polynomials $P$ on $n$ variables of degree $2$, such that any depth-$3$ arithmetic circuit computing $P$ must have top fan-in at least $\Omega(n)$. 
 
As an application, we get a new tradeoff between the top fan-in and formal degree in an approximate analog of the celebrated depth reduction result of Gupta, Kamath, Kayal and Saptharishi~\cite{gkks13b}. Formally, we show that if a degree $d$ homogeneous  polynomial $P$ can be computed by an arithmetic circuit of size $s$, then for every $t \leq d$, $P$ is in the border of a depth-$3$ circuit of top fan-in $s^{O(t)}$ and formal degree $s^{O(d/t)}$. To the best of our knowledge, the upper bound on the top fan-in in the original proof of~\cite{gkks13b} is always at least $s^{O(\sqrt{d})}$, regardless of the formal degree.
\end{abstract}

\thispagestyle{empty}
\pagenumbering{arabic}
\section{Introduction}
An arithmetic circuit on variables $\vx = (x_1, x_2, \ldots, x_n)$ over a field $\F$ is a directed acyclic graph with leaves labeled by variables in $\vx$ and constants in $\F$ and internal vertices labeled by sum $(+)$ or product $(\times)$. Such a circuit provides a natural succinct representation for multivariate polynomials in $\F[\vx]$. In this paper, the principle object of interest will be arithmetic circuits of depth-$3$, which we now define.
 
A depth-$3$  circuit (denoted by $\sps$) is an arithmetic circuit whose internal gates are arranged in three layers of alternating sums and products with the top layer being a sum. Such a circuit gives a representation of a polynomial as a sum of products of affine forms. 
Two of the parameters of interest of a $\sps$ circuit $C$ are the fan-in of the top layer, which we call the top fan-in of $C$ and the maximum of the degrees of its product gates, which we call the formal degree of $C$. 

A crucial point to note is that the formal degree of a circuit $C$ can be much much larger than the degree of the polynomial computed by $C$. Such a $\sps$ circuit computes polynomials of really high degree by taking products of affine forms, and then takes a linear combination of many such high degree polynomials to \emph{efficiently} compute a much lower degree polynomial via cancellations. One classic example of this is a result of Ben-Or~\cite{nw1997} who showed that over large enough fields, for every degree $d \in [n]$, the elementary symmetric polynomial of degree $d$ in variables $\vx$ defined as 
\[
\sym_d(\vx) = \sum_{S \in \binom{[n]}{d}} \prod_{j \in S} x_j \, , 
\]
can be computed by a $\sps$ circuit with top fan-in $n+1$ and formal degree $n$. In sharp contrast,  Nisan and Wigderson~\cite{nw1997} had earlier shown that any $\sps$ circuit computing $\sym_d(\vx)$ of formal degree at most $O(d)$ must have top fan-in at least $n^{\Omega(d)}$ (for a very large range of choice of $d$). Thus, higher formal degree indeed helps make the computation efficient. 

Another recent example of the power non-homogeneous depth-$3$ circuits is a beautiful result of Gupta, Kamath, Kayal and Saptharishi~\cite{gkks13b}, who showed that over the field $\Comp$, there is a $\sps$ circuit of formal degree $\exp(O(\sqrt{n}\log n))$ and top fan-in $\exp(O(\sqrt{n}\log n))$ which computes the determinant of an $n\times n$ symbolic matrix. Prior to this work, the best $\sps$ circuit known for the determinant had size $\exp(\Omega(n\log n))$. In fact, in~\cite{gkks13b}, the authors prove something stronger. They show that over the field $\Comp$, if a homogeneous $n$-variate polynomial of degree $d$ can be computed by an arithmetic circuit of size $s$, then it can be computed by an $\sps$ circuit of top fan-in $\exp(O(\sqrt{d}\log s))$ and formal degree $\exp(O(\sqrt{d}\log s))$. 

Thus, increasing the the formal degree of $\sps$ circuits helps reduce their top fan-in (and size) while preserving the expressiveness. The main motivation for this work is to explore this tradeoff between the formal degree and top fan-in better. In particular, the following question is of interest to us. 
\begin{question}\label{q:1}
In the depth reduction results of~\cite{gkks13b}, can we increase the formal degree of the resulting $\sps$ circuit further and obtain an $\sps$ circuit with a smaller top fan-in ? 
\end{question}
To the best of our understanding, the upper bound on the top fan-in of the $\sps$ circuit 
obtained by depth reducing a general circuit of size $s$ computing a degree $d$ polynomial is $s^{O(t + d/t)}$, and the formal degree is $s^{O(t)}$. Here, $t \in [d]$. Thus, regardless of the choice of $t$ and the resulting formal degree obtained, the top fan-in upper bound is always $s^{\Omega(\sqrt{d})}$. Before we state our results, we make a brief tour into the realm of approximative or border algebraic computation, since the notion is crucial to our results here.  
\subsection{Approximative or Border Complexity}
Let $\Phi : \F[\vx] \rightarrow \N$ be a measure of complexity of multivariate polynomials with respect to any reasonable model of computation.  For instance, we can think of $\Phi(P)$ 
to be circuit complexity, formula complexity, or depth-$3$ circuit complexity of $P$. We say that $P$ has border complexity with respect to $\Phi$ at most $s$ (denoted as $\overline{\Phi}(P) \leq s$) iff there are polynomials $Q_0, Q_1, Q_2, \ldots, Q_t \in \Comp[\epsilon, \vx]$ such that the polynomial $Q \equiv P + \sum_{j = 1}^t \epsilon^j\cdot Q_j$ satisfies $\Phi(Q) \leq s$. For brevity, we say that the polynomial $Q$ \emph{approximates} the polynomial $P$. This notion of approximation is often referred to as the \emph{algebraic} approximation, as opposed to the usual \emph{topological} notion. For this paper, we only work with algebraic approximation upper bounds, which trivially imply upper bounds in the topological sense as well. We refer the reader to an excellent discussion on border complexity in the work of Bringmann et al.~\cite{BIZ17}. 

It follows from the definitions that $\Phi(P) \leq s$ trivially implies $\overline{\Phi}(P) \leq s$, but in general we do not know implications in the other direction. In particular, it is potentially easier to prove border complexity upper bounds for a model, and harder to prove border complexity lower bounds. And, indeed we know some upper bounds in the border complexity framework which are either false, or not known in the exact computation framework. We state two such results. 

\paragraph*{Low degree factors of polynomials with small circuits. }In~\cite{B04}, B{\"{u}}rgisser showed that if a polynomial $P \in \Comp[\vx]$ has circuit complexity at most $s$, and $f$ is an irreducible factor of $P$ of degree $d$, then $f$ has border circuit complexity at most $\poly(s, d)$. In particular, this upper bound does not depend on the degree of $P$ itself, which could be as large as $2^s$! For exact computation, we only know that $f$ can be computed by  an arithmetic circuit of size $\poly(s, d, m)$, where $m$ is the largest integer such that $f^m $ divides $P$. Note that $m$ can be as large as $\exp(\Omega(s))$, and hence the bound is not always polynomially bounded in $s$ and $d$. Extending the results in~\cite{B04} continues to be a fascinating and fundamental open problem.

\paragraph*{Width-$2$ algebraic branching programs. }
 In~\cite{BIZ17}, Bringmann, Ikenmeyer and Zuiddam showed that over all fields of characteristic different from $2$, if  a polynomial $P$ of degree $d$ has an arithmetic formula of size $s$, then $P$ is in the border of width two algebraic branching programs of size at most $\poly(s)$. This is the approximative version of a strengthening of a classical result of Ben-Or and Cleve~\cite{bc88} who showed that if a degree $d$ polynomial $P$ has an arithmetic formula of size $s$, then $P$ can be computed by a width-$3$ algebraic branching program of size $\poly(s)$. 
 The result in~\cite{BIZ17} is surprising, since we know that an analog of the result of Ben-Or and Cleve is false for width-$2$ algebraic branching programs. In fact, Allender and Wang~\cite{AW16} showed that width-$2$ algebraic branching programs are not even complete, i.e. there are polynomials which they cannot compute regardless of the size. 
 
\subsection{Results}
Our main result is the following theorem. 
\begin{theorem}[Approximating general polynomials]\label{thm:approx general}
Let $P\in \Comp[\vx]$ be any homogeneous polynomial of degree $d$. Then, there exists a $\sps$ circuit $C \in \Comp(\epsilon)[\vx]$ with top fan-in at most $d+1$ and formal degree at most $\left(d\cdot 2^d \cdot \binom{n+d-1}{d-1}\right)$ such that \[
C \equiv P + \epsilon Q\, ,
\] where $Q \in \Comp[\epsilon, \vx]$  and every monomial with a non-zero coefficient in $Q$ has degree strictly greater than $d$.  
\end{theorem}

Thus, \emph{every} homogeneous polynomial of degree $d$ is in the border of $\sps$ circuits with top fan-in at most $d+1$. Of course, the upper bound on the formal degree is extremely high, and up to lower order terms, this is unavoidable due to counting arguments. We remark that this result is a bit surprising since it known to be false in the realm of exact computation. More formally, the following folklore lemma is well known (at least implicitly).
\begin{lemma}[Follows from Lemma 4.9 in~\cite{sw2001}, Lemma A.1 in~\cite{CGJW016}]
Any $\sps$ circuit computing the inner product polynomial $IP = \sum_{i = 1}^n x_iy_i$ must have top fan-in $\Omega(n)$, regardless of the formal degree. 
\end{lemma}
 Thus, the exact computation analog of~\autoref{thm:approx general} is false in a very strong sense, even for polynomials of degree $2$. 
 
An interesting implication of~\autoref{thm:approx general} is that one cannot hope to prove super linear top fan-in lower bound for $\sps$ circuits for polynomials of degree $O(n)$ without taking into account the formal degree, provided the lower bound also applies to the border of $\sps$ circuits. Most of the known lower bound results for arithmetic circuits do in fact extend to the border of the corresponding complexity class. 

Our second theorem is a special case of~\autoref{thm:approx sum of powers} for sums of powers of linear forms. The upper bound on the formal degree of the approximating $\sps$ circuit obtained here is much better.  
\begin{theorem}[Approximating sums of powers of linear forms]\label{thm:approx sum of powers}
Let $P = \sum_{i = 1}^T \ell_i^d$ be any homogeneous polynomial of degree $d$ in $\Comp[\vx]$, where each $\ell_i$ is a homogeneous linear form. Then, there exists a $\sps$ circuit $C \in \Comp(\epsilon)[\vx]$ with top fan-in at most $d+1$ and formal degree at most $\left(d\cdot T\right)$ such that \[
C \equiv P + \epsilon Q\, ,
\] where $Q \in \Comp[\epsilon, \vx]$  and every monomial with a non-zero coefficient in $Q$ has degree strictly greater than $d$.  
\end{theorem}
Our final result answers~\autoref{q:1} in the affirmative in the border complexity sense. We prove the following statement.
\begin{theorem}[Top fan-in vs formal degree for chasm at depth-$3$]\label{thm:approx d3 chasm}
Let $P \in \Comp[\vx]$ be a homogeneous polynomial of degree $d$ which is computable by an arithmetic circuit of size $s$. Then, for every $t \in [d]$, there is a $\sps$ circuit $C_{t} \in \Comp(\epsilon)[\vx]$ of top fan-in at most $s^{O(t)}$ and formal degree $s^{O(d/t)}$ such that 
\[
C_{t} \equiv P + \epsilon Q\, ,
\] where $Q \in \Comp[\epsilon, \vx]$  and every monomial with a non-zero coefficient in $Q$ has degree strictly greater than $d$.  
\end{theorem}
As remarked earlier, this tradeoff is in contrast to the original result of Gupta, Kamath, Kayal and Saptharishi~\cite{gkks13b} where the top fan-in is always at least $s^{O(\sqrt{d})}$ regardless of the formal degree of the circuit. 

In the rest of this note, we include the proofs of the above theorems. All our proofs are based on very simple and elementary ideas building on top of known results in this area, most notably those in~\cite{shp02, gkks13b}. However, the theorem statements, and in particular~\autoref{thm:approx general} seems to be interesting (and surprising!).
\section{Preliminaries}
\begin{itemize}
\item Unless otherwise stated, we work over the field   $\Comp$ of complex numbers. 
\item $n$ is the number of variables and $d$ is the degree. 
\item We use boldface letters like $\vx$ to denote the set of variables $\set{x_1, x_2, \ldots, x_n}$. 
\item For a natural number $m > 0$,  $[m]$ denotes the set $\set{1,\dots, m}$. 
\item For a scalar $\alpha \in \F$, and a polynomial $P \in \F[\vx]$, $P(\alpha \cdot \vx) = P(\alpha\cdot x_1, \alpha \cdot x_2, \ldots, \alpha \cdot x_n )$.
\end{itemize}

\begin{theorem}[Shpilka, Theorem 2.1 in~\cite{shp02}]\label{thm:shp general}
Let $P\in \Comp[\vx]$ be any homogeneous polynomial of degree $d$. Then, there exist homogeneous linear forms $\ell_1, \ell_2, \ldots, \ell_m \in \Comp[\vx]$ for $m \leq \left(d\cdot 2^d \cdot \text{Sparsity}(P) \right)$, such that 
\[
P = \sym_d(\ell_1, \ell_2, \ldots, \ell_m) \, .
\]
\end{theorem}

\begin{theorem}[Shpilka, Lemma 2.4 in~\cite{shp02}]\label{thm:shp sum of powers}
Let $P = \sum_{i = 1}^T \ell_i^d$ be any homogeneous polynomial of degree $d$ in $\Comp[\vx]$, where each $\ell_i$ is a homogeneous linear form. Let $\omega$ be a primitive root of unity of order $d$. Then, 
\[
P = -\sym_d(-\ell_1, -\omega \ell_1, -\omega^2 \ell_1,  \ldots, -\omega^{d-1}\ell_1, -\ell_2, \omega \ell_2, \ldots, -\omega^{d-1} \ell_T) \, .
\]
\end{theorem}

\section{Proofs and technical details}

\begin{lemma}[Approximating low degree homogeneous components]\label{lem:extracting homog components}
Let $\F$ be any field of size at least $d+1$ and let $P(\vx) = \sum_{i = 0}^d P_i(\vx)$ be a polynomial of degree $d$ in $\F[\vx]$ where for each $i \in \set{0, 1, \ldots, d}$, $P_i(\vx)$ is the homogeneous component of $P$ of degree equal to $i$. Then, for every $i \in \set{0, 1, \ldots, d}$ and for every choice of $i+1$ distinct elements $\alpha_{i,0}, \alpha_{i, 1}, \ldots, \alpha_{i, i}$ in $\F$, there exist $\beta_{i,0}, \beta_{i, 1}, \ldots, \beta_{i, i}$ in $\F$ such that 
\[
\sum_{j = 0}^i \beta_{i, j}\cdot  P(\alpha_{i, j} \cdot \vx)  = P_i(\vx) +  R \, , 
\]
where the degree of every monomial with a non-zero coefficient in $R$ is at least $i + 1$.
\end{lemma}
\begin{proof}
Let $y$ be a new formal variable, and let $Q(y) \in (\F[\vx])[y]$ be the defined as 
\[
Q(y) = P(yx_1, yx_2, \ldots, yx_n) \, .
\]
Clearly, $Q(y) = \sum_{j = 0}^d y^j P_j(\vx)$. For the rest of the proof, we fix an arbitrary $i \in \set{0, 1 , \ldots, d}$. Let $\alpha_{i, 0}, \alpha_{i, 1}, \ldots, \alpha_{i, i}$ be any $i+1$ distinct elements of $\F$. Then, for every $k \in \set{0, 1, \ldots, j}$, we have 
\[
Q(\alpha_{i, k}) = \sum_{j = 0}^d \alpha_{i, k}^j P_j(\vx) \, .
\]
For $j \in \set{0, 1, \ldots, i}$, let $\gamma_j = (\alpha_{i, j}^0, \alpha_{i, j}^{1}, \alpha_{i,j}^2, \ldots, \alpha_{i, j}^i)$. From the choice of $\alpha_{i, j}$, we know that for any $j \neq j'$, $\alpha_{i, j} \neq \alpha_{i, j'}$. Thus, it follows that $\gamma_0, \gamma_1, \ldots, \gamma_i$ are linearly independent, and hence there exist scalars $\beta_{i, 0}, \beta_{i, 1}, \ldots, \beta_{i, i}$ in $\F$ such that 
\begin{equation}\label{eqn:1}
\sum_{j = 0}^i \beta_{i, j} \gamma_j = (0, 0, \ldots, 1) \, .
\end{equation}
Therefore, 
\begin{align*}
\sum_{j = 0}^i \beta_{i, j} Q(\alpha_{i, j}) & = \sum_{j = 0}^i \beta_{i, j} \left(\sum_{j' = 0}^d \alpha_{i, j}^{j'} P_{j'}(\vx) \right)\\
&= \sum_{j' = 0}^d \left( \sum_{j = 0}^i \beta_{i, j}  \alpha_{i, j}^{j'} \right) P_{j'}(\vx) 
\end{align*}
From~\autoref{eqn:1}, we know that for every $j' < i$, 
\[
\sum_{j = 0}^i \beta_{i, j} \alpha_{i, j}^{j'} = 0 \, , 
\]
and, 
\[
\sum_{j = 0}^i \beta_{i, j} \alpha_{i, j}^i = 1 \, . 
\]
Thus, 
\begin{align*}
\sum_{j = 0}^i \beta_{i, j} Q(\alpha_{i, j}) &= P_i(\vx) +\left(\text{monomials of degree} > i\right) \, .
\end{align*}
\end{proof}

\begin{lemma}[Approximating powers of diagonal circuits]\label{lem:approx duality}
Let $\ell_1, \ell_2, \ldots, \ell_k \in \Comp[\vx]$ be homogeneous linear forms and let $P = \left(\sum_{j = 1}^k \ell_i^b \right)^a$ be a homogeneous polynomial of degree $ab$ for any $a, b$.  Then, there exists a $\sum\prod\sum$ circuit $C \in \Comp(\epsilon)[\vx]$ of top fan-in at most $a+1$ and formal degree at most $O(kab)$ such that 
\[
C \equiv P + \epsilon Q \, , 
\] 
 where $Q \in \Comp[\epsilon, \vx]$  and every monomial with a non-zero coefficient in $Q$ has degree at least $d+1$.  
 \end{lemma}
\begin{proof}
The proof follows the proof of the original duality lemma of Saxena~\cite{sax08} except in place of the interpolation used there (which blows up the top fan-in by a factor as large as $k$), we only do a partial interpolation in the spirit of~\autoref{lem:extracting homog components} here, and this only incurs a small blow up in the top fan-in. We now sketch some of the details. 

By $\exp(y)$, we denote the formal power series $1 + y + \frac{y^2}{2!} + \frac{y^3}{3!} + \ldots$, and by $E_a(y)$ we denote the truncation of this power series to monomials of degree at most $a$, i.e. $E_a(y) = 1 + y + \frac{y^2}{2!} + \frac{y^3}{3!} + \ldots + \frac{y^{a}}{a!}$. Now, observe that 
\begin{eqnarray}
\frac{1}{a!}  \cdot P &=& \cf_{y^a}\left[\exp\left( y\cdot  \left(\sum_{j = 1}^k \ell_j^b \right)\right) \right] \, , \\
&=& \cf_{y^a}\left[\prod_{j = 1}^k E_a\left( y\cdot \ell_j^b \right) \right] \, ,
\end{eqnarray}
We are now in a scenario similar to what we handled in the proof of~\autoref{lem:extracting homog components}; we have a polynomial $R(y) = \left[\prod_{j = 1}^k E_a\left( y\cdot \ell_j^b \right) \right] $ which we think of as univariate in $y$, and we are interested in the coefficient of $y^a$ in this polynomial which has degree $ak$. From the proof of~\autoref{lem:extracting homog components}, we know that for any $a+1$ distinct field elements $\alpha_1, \alpha_2, \ldots, \alpha_{a+1}$, there exist field elements $\beta_1, \beta_2, \ldots, \beta_{a+1}$ such that 
\[
\sum_{u = 1}^{a+1} \beta_u \cdot R(\alpha_u) = \cf_{y^a}[R(y)] + Q \, ,
\]
where $Q$ is some linear combination of monomials of degree at least $a+1$. Also, note that for every $\alpha_u \in \Comp$, $R(\alpha_u)$ can be written as a product of linear functions. Thus, $\sum_{u = 1}^{a+1} \beta_u \cdot R(\alpha_u)$ can be computed by a $\sps$ circuit $\tilde{C}$ with top fan-in $a+1$ and formal degree $O(abk)$. Replacing every $x_i$ in $\tilde{C}$ by $\epsilon x_i$ and dividing by $\epsilon^d$ completes the proof of the lemma.
\end{proof}

\paragraph{Approximating general homogeneous polynomials. }
\begin{proof}[Proof of~\autoref{thm:approx general}]
Let $P \in \Comp[\vx]$ be any homogeneous degree $d$ polynomial. Then, from~\autoref{thm:shp general}, we know that for $m \leq \left(d\cdot 2^d \cdot \text{Sparsity}(P) \right)$ there are homogeneous linear forms $\ell_1, \ell_2, \ldots, \ell_m$ such that 
\[
P = \sym_d(\ell_1, \ell_2, \ldots, \ell_m) \, .
\]
It was observed by Michael Ben-Or (see~\cite{nw1997})\footnote{and, is easy to see} that $P$ equals the homogeneous component of $\prod_{i = 1}^m (\ell_i + 1)$ of degree equal to $d$. Thus, from~\autoref{lem:extracting homog components}, we know that there exist scalars $\alpha_0, \alpha_1, \ldots, \alpha_d$ and $\beta_0, \beta_1, \ldots, \beta_d$ such that 
\[
P(\vx) + \left(\text{monomials of degree} > d \right) = \sum_{j = 0}^d \beta_j \left( \prod_{i = 1}^m (\alpha_ j \cdot \ell_i + 1)\right) \, .
\]

Replacing every $x_i$ in $\tilde{C}$ by $\epsilon x_i$ and dividing by $\epsilon^d$ completes the proof of the theorem.

\end{proof}

\paragraph{Approximating sums of powers of linear forms. }
\begin{proof}[Proof of~\autoref{thm:approx sum of powers}]
The proof of~\autoref{thm:approx sum of powers} is essentially the same as the proof of~\autoref{thm:approx general}, except for the fact that we use~\autoref{thm:shp sum of powers} in the place of~\autoref{thm:shp general}, and this gives us the desired  upper bound on formal degree of $P$ in terms of its Waring rank, as opposed to  sparsity. 
\end{proof}

\paragraph{Approximate reduction to depth-$3$. }
We present two proofs of~\autoref{thm:approx d3 chasm}, both simple, but one simpler than the other. 
\begin{proof}[First proof of~\autoref{thm:approx d3 chasm}]
Let $a \in [d]$ be a parameter. As a first step, we use the standard depth reduction to depth-$4$~\cite{av08, Tav15} to transform a homogeneous circuit $C$ of unbounded depth to a homogeneous $\sum^{[T]}\prod^{[O(d/a)]}\sum^{[M]}\prod^{[a]}$ circuit $C'$, with $T = s^{O(\frac{d}{a})}$ and $M = n^{O(a)}$. Now, for every $\sum^{[M]}\prod^{[a]}$ sub-circuit at the bottom two levels, we apply~\autoref{thm:approx general}, which says that each of these sub-circuits can be \emph{approximated} by $\sum^{[a+1]}\prod^{[n^{O(a)}]}\sum$ circuits. We now plug these depth-$3$ circuits into $C'$ to get a $\sum^{[T]}\prod^{[O(d/a)]}\sum^{[a+1]}\prod^{[n^{O(a)}]}\sum$ circuit $C''$. Expanding out the product gates at the second level by brute force, we obtain a $\sum^{[T]}\sum^{[a^{O(d/a)}]}\prod^{[O(d/a)\cdot n^{O(a)}]}\sum$ circuit which approximates the polynomial computed by $C$. Combining the sum gates in the first two layers gives us the desired depth-$3$ circuit. It is also not hard to see that the polynomial computed by the resulting depth-$3$ circuit approximates the polynomial computed by $C$ in the sense of the statement of~\autoref{thm:approx d3 chasm}. We skip the details.  
\end{proof}

\begin{proof}[Second proof of~\autoref{thm:approx d3 chasm}]
For this proof, we will follow the outline in~\cite{gkks13b}, with one difference. In place of the use of the duality lemma of Saxena to transform a homogeneous depth-$5$ powering circuit to a non-homogeneous depth-$3$ circuit, we use~\autoref{lem:approx duality}. We now elaborate on some of the details, and chase the parameters in the process. We refer the reader to the original paper of Gupta et al~\cite{gkks13b} for details. 

Let $a \in [d]$ be a parameter. We first transform a homogeneous circuit of unbounded depth to a homogeneous $\sum^{[T]}\bigwedge^{[a]}\sum^{[M]}\bigwedge^{[d/a]}\sum$ circuit using the standard depth reduction  to depth-$4$~\cite{av08, Tav15} and Fischer's formula (Lemma IV.3 in~\cite{gkks13b}). It follows from the proofs of Lemma IV.4 and Lemma IV.2 in~\cite{gkks13b} that $T \leq s^{O(a)}$ and $M \leq s^{O(d/a)}$. 

Now, we apply~\autoref{lem:approx duality} to each $\bigwedge^{[a]}\sum^{[M]}\bigwedge^{[d/a]}\sum$ sub-circuit, and take their sum. This gives us a $\sps$ circuit 
with top fan-in $O(Td)$ and formal degree $O(Md)$. This completes the proof.
\end{proof}

\section{Open problems}
We end with some open problems. 
\begin{itemize}
\item Perhaps the most natural question is to understand if the exact computation versions of~\autoref{thm:approx general}, \autoref{thm:approx sum of powers} or \autoref{thm:approx d3 chasm} are true with a reasonable blow up in the parameters. For instance, can every homogeneous polynomial of degree $d$ be computed by a $\sps$ circuit with top fan-in $\poly(n)$, is arbitrarily large formal degree is allowed ? What about all polynomials in $\VP$ ?
\item Another question is to understand if there are natural classes of polynomials for which the formal degree upper bound in~\autoref{thm:approx general} can be reduced to some more reasonable (and possibly useful) upper bound, while keeping the top fan-in small. For instance, can every homogeneous degree $d$ polynomial with a formula of size $s$ be approximated by a $\sps$ circuit of top fan-in $\poly(d)$ and formal degree $\poly(s,d)$ ?
\item And finally,~\autoref{thm:approx general} seems to be a very general structural statement  for low degree polynomials. Does this have other applications ?
\end{itemize}

\paragraph*{Acknowledgements. } I am thankful to Michael Forbes, Amir Shpilka and Ramprasad Saptharishi for helpful discussions.  
\bibliographystyle{customurlbst/alphaurlpp}
\bibliography{references}

\newcommand{\etalchar}[1]{$^{#1}$}
\begin{thebibliography}{GKKS13}

\bibitem[AV08]{av08}
Manindra Agrawal and V.~Vinay.
\newblock \href {http://dx.doi.org/10.1109/FOCS.2008.32} {Arithmetic Circuits:
  A Chasm at Depth Four}.
\newblock In {\em \FOCS{2008}}, pages 67--75, 2008.
\newblock Pre-print available at \parseECCC{TR08/062}.

\bibitem[AW16]{AW16}
Eric Allender and Fengming Wang.
\newblock \href {http://dx.doi.org/10.1007/s00037-015-0114-7} {On the power of
  algebraic branching programs of width two}.
\newblock {\em computational complexity}, 25(1):217--253, 2016.

\bibitem[BC88]{bc88}
Michael {Ben-Or} and Richard Cleve.
\newblock {Computing Algebraic Formulas Using a Constant Number of Registers}.
\newblock In {\em \STOC{1988}}, pages 254--257, 1988.

\bibitem[BIZ17]{BIZ17}
Karl Bringmann, Christian Ikenmeyer, and Jeroen Zuiddam.
\newblock \href {http://dx.doi.org/10.4230/LIPIcs.CCC.2017.20} {On Algebraic
  Branching Programs of Small Width}.
\newblock In {\em 32nd Computational Complexity Conference, {CCC} 2017, July
  6-9, 2017, Riga, Latvia}, pages 20:1--20:31, 2017.

\bibitem[B{\"{u}}r04]{B04}
Peter B{\"{u}}rgisser.
\newblock \href {http://dx.doi.org/10.1007/s10208-002-0059-5} {The Complexity
  of Factors of Multivariate Polynomials}.
\newblock {\em Foundations of Computational Mathematics}, 4(4):369--396, 2004.

\bibitem[CGJ{\etalchar{+}}16]{CGJW016}
Mahdi Cheraghchi, Elena Grigorescu, Brendan Juba, Karl Wimmer, and Ning Xie.
\newblock \href {http://dx.doi.org/10.4230/LIPIcs.ICALP.2016.35} {AC{\^{}}0 o
  MOD{\_}2 Lower Bounds for the Boolean Inner Product}.
\newblock In {\em 43rd International Colloquium on Automata, Languages, and
  Programming, {ICALP} 2016, July 11-15, 2016, Rome, Italy}, pages 35:1--35:14,
  2016.

\bibitem[GKKS13]{gkks13b}
Ankit Gupta, Pritish Kamath, Neeraj Kayal, and Ramprasad Saptharishi.
\newblock \href {http://dx.doi.org/10.1109/FOCS.2013.68} {{Arithmetic Circuits:
  {A} Chasm at Depth Three}}.
\newblock In {\em \FOCS{2013}}, pages 578--587, 2013.
\newblock Pre-print available at \parseECCC{TR13/026}.

\bibitem[NW97]{nw1997}
Noam Nisan and Avi Wigderson.
\newblock \href {http://dx.doi.org/10.1007/BF01294256} {Lower bounds on
  arithmetic circuits via partial derivatives}.
\newblock {\em Computational Complexity}, 6(3):217--234, 1997.
\newblock Available on
  \href{http://citeseerx.ist.psu.edu/viewdoc/summary?doi=10.1.1.90.2644}{\tt
  citeseer:10.1.1.90.2644}.

\bibitem[Sax08]{sax08}
Nitin Saxena.
\newblock \href {http://dx.doi.org/10.1007/978-3-540-70575-8_6} {{Diagonal
  Circuit Identity Testing and Lower Bounds}}.
\newblock In {\em \ICALP{2008}}, pages 60--71, 2008.
\newblock Pre-print available at \parseECCC{TR07/124}.

\bibitem[Shp02]{shp02}
Amir Shpilka.
\newblock \href
  {http://dx.doi.org/https://doi.org/10.1016/S0022-0000(02)00021-1} {Affine
  projections of symmetric polynomials}.
\newblock {\em Journal of Computer and System Sciences}, 65(4):639 -- 659,
  2002.
\newblock Special Issue on Complexity 2001.

\bibitem[SW01]{sw2001}
Amir Shpilka and Avi Wigderson.
\newblock \href {http://dx.doi.org/10.1007/PL00001609} {Depth-3 arithmetic
  circuits over fields of characteristic zero}.
\newblock {\em Computational Complexity}, 10(1):1--27, 2001.
\newblock \pCCC{1999}.

\bibitem[Tav15]{Tav15}
S{\'{e}}bastien Tavenas.
\newblock \href {http://dx.doi.org/10.1016/j.ic.2014.09.004} {Improved bounds
  for reduction to depth 4 and depth 3}.
\newblock {\em Inf. Comput.}, 240:2--11, 2015.
\newblock \pMFCS{2013}.

\end{thebibliography}

\end{document}